\documentclass[12pt]{article}%

\usepackage{amsmath, amsthm, amssymb, mathtools}
 \usepackage{graphicx, subfigure}

\usepackage[colorlinks, citecolor=blue, linkcolor=blue]{hyperref}
\usepackage{natbib}
\usepackage{fullpage}

\usepackage{color}

\usepackage{setspace}

\usepackage{csquotes}

\usepackage{relsize}

\newcommand{\is}{$i$'s }
\newcommand{\js}{$j$'s }

\newcommand{\sq}{\ensuremath{\subseteq}}

\newtheorem{Definition}{Definition}

\newtheorem{Theorem}{Theorem}
\newtheorem{Proposition}{Proposition}

\newtheorem{Corollary}{Corollary}

\newtheorem{Example}{Example}

\usepackage{eurosym}

\begin{document}

\title{No Trade Under Verifiable Information\thanks{I would like to thank Arz\'e Karam, Jan Christoph Schlegel, an anonymous referee, and participants at the EC24 Workshop on Blockchains and Decentralized Finance at Yale, for useful comments. I would like to acknowledge financial support from the Economic and Social Research Council (ES/V004425/1). The previous title of this paper was ``No Trade in a Blockchain". Declarations of interest: none.}\\ 
\vspace{0.9cm}
\author{Spyros Galanis\thanks{Department of Economics, University of Durham, spyros.galanis@durham.ac.uk.}}}

\maketitle

\begin{abstract}
No trade theorems examine conditions under which agents cannot agree to disagree on the value of a security which pays according to some state of nature, thus preventing any mutual agreement to trade. A large literature has examined conditions which imply no trade, such as relaxing the common prior and common knowledge assumptions, as well as allowing for agents who are boundedly rational or ambiguity averse. We contribute to this literature by examining conditions on the private information of agents that reveals, or verifies, the true value of the security. We argue that these conditions can offer insights in three different settings: insider trading, the connection of low liquidity in markets with no trade, and trading using public blockchains and oracles. 
%The current paper extends this result by stud on common priors  We consider condition on verifiable information. In a decentralised and anonymous environment, such as a blockchain,  a trade is settled when the security's value is verified by an oracle, an intermediary who knows the state of nature. However, since the identity of agents in a blockchain is hidden, an oracle can impersonate a agent, by participating in a trade only when he knows that the outcome will be favourable. If other agents know that this is possible, they might be unwilling to participate in any trade that is offered. In this paper, we examine conditions on the private information of agents which are necessary and sufficient for a no trade result.

% In other words, anonymity of agents is incompatible with  verifiability of trades. We find that the requirement on verifiability is relatively mild. In particular, it is enough that,
%%  In an environment with different prior beliefs, we find that the following condition of partial verification of the security is sufficient for no trade and information aggregation.
% among all values of the security that it is common knowledge that some agent thinks are possible, there exists an oracle that can verify either the maximum or the minimum. 

\end{abstract} 

%\strut

 \vspace{0.2cm}

%\strut


\section{Introduction}
\label{introduction}

%Contrary to popular intuition, asymmetric information alone cannot explain high trading volumes in financial markets. 
Contrary to common belief, asymmetric information by itself is insufficient to account for the high trading volumes observed in financial markets.
\cite{aum76} first showed that if agents share a common prior, then it cannot be common knowledge that they disagree on the probability they assign to an event.  \cite{milSto82} use this result to show that if an allocation is ex ante Pareto efficient, then new asymmetric information will not lead to trade if agents are risk averse and hold concordant beliefs. There is now a large literature  examining conditions that generate no trade, such as relaxing the common prior or common knowledge assumptions, or having agents who are boundedly rational or ambiguity averse.

The current paper contributes to this literature by examining a different set of conditions that lead to no trade, 
%and in particular conditions on the private information of agents which result in revealing (or verifying) the true value of the traded security.
specifically examining how agents’ private information can reveal or verify the true value of the security being traded.
%examining when it is the case that knowledge about the value of the traded security generates a no trade result. 
To provide an example, suppose that it is common knowledge that agent $i$ always knows the true state and therefore the true value of the traded security. Then, no other agent would be willing to accept a buy or a sell order and there would be no trade, irrespective of whether the common prior assumption holds. The reason is that it cannot be common knowledge that we disagree on the value of the security, when it is common knowledge that agent $i$ knows its value. Hence, accepting a trade `must' be wrong if you are not agent $i$.  Is it possible to further weaken this condition regarding the verifiability of the security’s value while still implying no trade?  For example, we could require that the conjunction of everyone's knowledge verifies the true value of the security. Alternatively, what if there is an agent who always knows whether the value of the security is above or below a certain threshold? Are these conditions necessary or sufficient to preclude trade? To motivate these questions, we provide three examples.

The first is insider trading. The UK Criminal Justice Act 1993 states in Section 52 that ‘An individual who has information as an insider is guilty of insider dealing if [...] he deals in securities that are price-affected securities in relation to the information’. This clause can be interpreted to mean that insider trading occurs when the insider possesses private information significant enough to affect the price—or underlying value—of the security. Put differently, the insider knows whether the value of the security lies above or below a certain threshold. 

Suppose that agents bet on the financial performance of a company, by trading a derivative on the price of its share. If an employee has insider information which always reveals the true value of the company, it is straightforward that she should not be allowed to trade. The reason is not that she can make huge profits by `taking advantage' of other agents; it is because, if agents are rational, they will understand that this possibility exists and there will be no trade. A no trade result is detrimental in this case because the  price of the derivative may not be sufficiently close to its true value, hence the market would fail in aggregating information. However, what if the employee's knowledge is not enough to verify the true value of the company, but only whether it is above a certain threshold? By understanding the conditions on verifiability that result in no trade, we can determine under which conditions we should allow insiders to participate in a financial market.

The second example links low liquidity with no trade. In 1998, the Long-Term Capital Management (LTCM) hedge fund was facing collapse due to the Russian debt default. The resulting fire sales and the costly disruptions for the world financial markets led the Federal Reserve Bank of New York to hold meetings to recapitalise LTCM. \cite{cai03} documents how market makers in the bonds futures market made huge profits using privileged information on customer order flow. What happens to trading if it becomes common knowledge that market insiders know which assets are distressed and therefore have a low value? If this leads to a no trade result, then the liquidity of the market will dry up, exacerbating the price impact of the insiders and the distressed sellers. \cite{cai03} shows, in fact,  that market liquidity was very low during this episode.%
\footnote{\cite{brunnermeier01} provides an overview of the literature on asymmetric information and stock market crashes, whereas \cite{brunnermeierPedersen05} construct a model of predatory trading with distressed sellers. \cite{karamBogoev23} provide an empirical investigation of the connection between low liquidity and momentum trading.}

The third example is trading on the blockchain. A public blockchain enables a decentralised trading environment, where completing a transaction does not require  a third party which custodies the agents' assets, verifies their ownership, or executes the transaction. All transactions are public and trades are executed algorithmically using {\it smart contracts}, when pre-specified events occur. 
%However, real-world events must be reported on the blockchain. 
For example, paying an option on a real-world asset requires the reporting of its value within the blockchain, so that the smart contract can be executed.
%The promise of  blockchain technology is that it can facilitate the trading of all kinds of securities, such as options, futures and stocks, in a decentralised environment, without relying on a third party which custodies the agents' assets, verifies their ownership, or executes the transaction. These trades are executed algorithmically using {\it smart contracts}, when pre-specified events occur. However, if the relevant event occurs outside the blockchain, then its realisation must be verified. For example, paying an option on a real-world asset requires the reporting of its value within the blockchain, so that the smart contract can be executed.
%
%purpose of this paper is to examine this tension, by studying conditions on the private information of individual agents that are necessary and sufficient for no trade. We therefore contribute to the literature which stems from \cite{aum76}, who showed that agents cannot agree to disagree. This literature has so far studied conditions on priors, common knowledge and bounded rationality to explain trade, but not on the private information of individual agents.   
In practice, the verification and reporting  of real-world events to the blockchain is performed by {\it oracles}.%
\footnote{\cite{garrattMonnet22} show that oracles are unavoidable, because the truthful reporting about the realisation of publicly observed events cannot be implemented as a unique equilibrium in a completely decentralized environment.}
%These trades are executed when some oracles know the true state of nature, or the value of the underlying security. 
However, because users of the blockchain are anonymous, nothing stops oracles from participating in the trade, especially when they know it is going to be favourable for them.%
\footnote{More accurately, they are pseudonymous, because the history of transactions is public and can be linked to the unique addresses.}
 Will this `insider trading' make trading on the blockchain inherently undesirable? 
% What matters is whether oracles learn the value of the security before or after the trade is agreed.  If they learn it after the trade has been agreed, for example when betting on the next Presidential election, then anonymity plays no role. If they learn it before, however, some agents may be reluctant to participate in any trades at all. An example is betting on the financial performance of a company, using a derivative on the price of its stock. An employee with insider information could play the role of the oracle, thus facilitating the trade, but also participate in it. Alternatively, an oracle can participate in a trade right after they learn the value of the asset and just before  they report it on the blockchain. In these cases, there is an inherent incompatibility between the anonymity of agents and the verifiability of the value of securities. The former is essential for  a decentralised system, whereas the latter is necessary for executing trades. 

%The purpose of this paper is to examine this tension, by studying conditions on the private information of individual agents that are necessary and sufficient for no trade. We therefore contribute to the literature which stems from \cite{aum76}, who showed that agents cannot agree to disagree. This literature has so far studied conditions on priors, common knowledge and bounded rationality to explain trade, but not on the private information of individual agents.  

We examine conditions on the verifiability of securities (that is, revealing their true value) using three different environments.  In all settings, there are finitely many states and agents, equipped with asymmetric information and different priors. Security $X$ pays a dividend according to the state of nature. Trades are agreed in the interim stage, after agents receive their private information but before the true state is revealed. The first environment is static and follows \cite{aum76}. We say that there is common knowledge trade at some state if  the agents' expected values of the security are common knowledge but different.  The second environment is dynamic with infinitely many periods, following \cite{geaPol82}. In each period, one agent announces (truthfully) his expected value of the security. All other agents update their information based on the announcement and form a new expected value. In the next period, another agent announces his expected value, and so on. We say that agents cannot disagree forever if they eventually agree on the expected value of the security. The third environment is also dynamic and follows \cite{ostrovsky12}. Agents are myopic and payoffs are determined using market scoring rules, which are most often used in prediction markets. We say that there is information aggregation if the price of the security always converges to its true value. 

What conditions on the private information of agents lead to no trade and information aggregation? On one extreme, we say that the security is {\it verifiable} if, for each state, there is some agent who knows its value in the interim stage, when the trade is agreed. This condition is strong and it  implies that there cannot be common knowledge trade. On the other extreme, the security is {\it collectively verifiable} if the conjunction of the knowledge of all agents reveals its true value. In this case, it is possible that no-one knows the value of the security, when the trade is agreed. We show by example that collective verifiability does not lead to no trade. 

Our main results involve two intermediate properties. We say that the security is {\it maxmin verifiable} if there exists a agent who knows whether the maximum (or minimum) has occurred, among all values of the security that it is common knowledge that are possible. It is {\it threshold verifiable} if there is a threshold $x$ within those values and a agent who knows at some state whether the security's value is strictly above or below $x$. Even though maxmin verifiability is significantly weaker than verifiability, we show that it is sufficient for no trade. Our main result, Theorem \ref{char}, shows that threshold verifiability is necessary and sufficient for no common knowledge trade. In a dynamic setting, however, a security that is threshold or maxmin verifiable at time $t$, may not be at time $t+1$. This means that these two properties are not sufficient for no trade or information aggregation, unless we assume that they hold at a distant time $t^{*}$, when agents have stopped updating their announcements.

\subsection{Related literature}
\label{literature}

%The \textit{no trade theorems} stem from \cite{aum76}, who showed that there cannot be common knowledge trade with common priors. 
%%
%\footnote{\cite{aum76} only analysed the case of agreeing on the posterior of an event $E$, however this is easily generalisable to the expectation of a random variable and common knowledge trade.}
%%
Several papers (\cite{mor94}, \cite{feinberg00}, \cite{samet98}, \cite{bonannoNehring99}, \cite{halpern02}, \cite{ng03}) characterize the existence of a common prior with respect to the condition that there does not exist a mutually beneficial trade that is common knowledge at all states. \cite{gea89} showed that trade can still occur with common priors if agents are boundedly rational and make information processing errors. These results were extended in an environment with unawareness by \cite{gal13, galanis18}, \cite{hms13} and \cite{meierSchipper14etbu}. The current paper differs from this literature because it imposes conditions on the private information of agents, rather than on priors or bounded rationality.

%\cite{geaPol82} was the first to study a dynamic version of information communication in a non-strategic setting.
%characterize the existence of a (not necessarily positive) common prior, with respect to the weaker
%condition that there does not exist a mutually beneficial trade that is common knowledge at {\it all} states.
%\cite{geaPol82}, \cite{cave83}, \cite{sebGea83}, \cite{nielsen84}, \cite{bacharach85} and \cite{nielsenEtAl90} study information communication in a non-strategic setting, where agents announce posterior beliefs or other aggregate statistics. 
%However, they do not fully characterize under what conditions the consensus yields the true posterior or expectation of the security, a gap which is filled by \cite{demarzoSkiadas98, demarzoSkiadas99}. 
%go beyond the consensus result of the former papers, finding necessary and sufficient conditions for information aggregation. 

\cite{ostrovsky12} and \cite{chenEtAl12} show that in a market with either myopic or strategic agents,  {\it separable} securities are both necessary and sufficient for information aggregation. Their models are based on market scoring rules (\cite{hanson03, hanson07}), hence their results are directly applicable to prediction markets (\cite{wolfersZitzewitz04}). Similar approaches can be found in \cite{dimitrovSami08}, \cite{galanisKotronis21}, \cite{galanisEtAl24} and \cite{galanisMikhalishchev25}, where the focus is on examining whether information gets aggregated under various assumptions regarding preferences and the signal structure, such as unawareness, ambiguity aversion, and costly information acquisition. 
%The information aggregation properties of separable securities are  studied  in an environment with unawareness by \cite{galanisKotronis19} and with ambiguity by \cite{galanisEtAl19}.
These papers identify classes of securities which ensure that information is aggregated. In the current model, our results apply to all securities, but the conditions on private information are imposed at a distant time $t^{*}$ where all information updating has concluded.

%\red{All the aforementioned papers study the effect on trade from common priors, common knowledge and how sophisticated the agents are in processing information. In contrast, this paper focuses on the effect of the knowledge that agents have at the time of agreeing to trade.}

The advantage of the market scoring rules (MSR) over more well-known market mechanisms, such as the continuous double auction, is that an agent can make her prediction/trade without waiting for another agent to take the opposite side, or submit a limit order and wait for it to be filled. This feature makes it an attractive mechanism for markets with relatively few participants who do not trade daily, or in markets with automated market makers. Automated market makers (AMMs) are widely used in Decentralized Finance, see \cite{schlegelEtAl22} for an axiomatization of the logarithmic MSR. For more on the connection between prediction markets and AMMs, see \cite{frongilloWaggoner17}, \cite{othmanSandholm11}, \cite{abernethyEtAl13}, and \cite{abernethyFrongillo12}.

There is a growing literature on the economics of blockchain and cryptocurrencies. \cite{biaisEtAl19} provides a game-theoretic analysis of the proof-of-work protocol, which is in Bitcoin used by the miners, who maintain and update the ledger. \cite{atheyEtAl16} and \cite{bohmeEtAl15} describe several empirical facts about the usage of Bitcoin, whereas \cite{chenEtAl19} examines the desirable properties of the proof-of-work protocol. \cite{easleyEtAl19} explains the emergence of transaction fees in Bitcoin using a game-theoretic analysis. \cite{budish18} argues that the payments to miners, for maintaining the blockchain, must be large relative to the one-off benefits of attacking it, thus making the blockchain vulnerable to attacks once it becomes valuable. \cite{schillingUhlig19} discusses the monetary policy implications when cryptocurrencies compete with traditional currencies, backed by a central bank. \cite{hinzenEtAl19} argues that the proof-of-work protocol is inherently unable to sustain a large volume of transactions. The reason is that a rise in transaction demand leads to an increase in block fees, attracting more miners and exacerbating network latency, thus delaying payment confirmation and making the payment platform less attractive for users. \cite{abadiBrunnermeier18} examines when record-keeping is better organised through a blockchain, instead of a traditional centralised intermediary.

In the current paper we assume that oracles do not act strategically and cannot lie about their reports, when they verify the value of the security. The issue of eliciting truthful reports in a strategic environment was first studied by \cite{prelec04} (Bayesian Truth Serum) and \cite{millerEtAl05} (peer-prediction mechanisms). In general, these mechanisms work by randomly pairing two agents and paying them according to how close their reports are. See \cite{jurcaBoi06}, \cite{witkowskiParkes12a, witkowskiParkes12b} for various extensions. \cite{goelEtAl19} uses the peer-prediction mechanism of \cite{radanovicEtAl16}, in order to study truthful reporting of oracles in a blockchain.

The paper is organised as follows. Section \ref{model} presents the basic model. Section \ref{trading environments} examines the three trading environments and presents the results on no trade. 
%Section \ref{discussion} discusses further issues that arise from our analysis.

\section{Model}
\label{model}
Let $I$ be a finite set of $n$ agents. Uncertainty is described by a finite set of states $\Omega$. Agent \is private information is represented by partition $\Pi_i$ and  prior $p_i$ that has full support on $\Omega$.
%\footnote{That is, $p_i$ assigns positive probability at all states in $\Omega$.}
%
We do not assume a common prior, so we allow $p_i \neq p_j$ for some $i,j \in I$.

Agent $i$ knows event $E$ at $\omega$ if $\Pi_i(\omega) \sq E$. This means that in all states that he considers possible at $\omega$, $E$ is true. Define $\Pi_i(F) = \underset{\omega' \in F}{\bigcup} \Pi_i(\omega')$ to be the set of all states that $i$  thinks are possible, if the true state is in $F$. Using this notation, we can say that $\Pi_j(\Pi_i(\omega))$ is the set of states that, at $\omega$, agent $i$ thinks that $j$ considers possible. If $\Pi_j(\Pi_i(\omega)) \sq F$, then we say that $i$ knows that $j$ knows $F$. An event $E$ is common knowledge at $\omega$ if  $\Pi_{i_n}(\Pi_{i_{n-1}} \ldots  (\Pi_{i_1}(\omega))) \sq E$, for any sequence of agents $i_1, \ldots, i_n$.%
\footnote{These notions are explained further in \cite{gea92}.}

Let $C(\omega)$ be the set of states that are reachable from $\omega$. Formally, $C(\omega)$ is the union of sets $\Pi_{i_n}(\Pi_{i_{n-1}} \ldots  (\Pi_{i_1}(\omega)))$, for any sequence of agents $i_1, \ldots, i_n$. Say that an event $E'$ is self-evident if, whenever it occurs, everyone knows it. Formally, for all $\omega' \in E'$ and $i \in I$, we have $\Pi_i(\omega') \sq E'$.Then, $C(\omega)$ can be described as the smallest self-evident event that contains $\omega$.%
\footnote{The partition generated by $C$ is called the finest common coarsening of the partitions of all agents.}
\cite{aum76} showed that an event $E$  is common knowledge at $\omega$ if and only if $C(\omega) \sq E$.

%in terms of a self-evident event $E'$, which whenever it occurs, everyone knows it. Formally, if $\omega' \in E'$, then $\Pi(\omega') \sq E'$ for all $i \in I$. We then have that event $E$ is common knowledge at $\omega$ if and only if it contains a self-evident event $E'$, so that $\omega \in E' \sq E$.
%An event $E \sq \Omega$ is common knowledge if everyone knows it, everyone knows that everyone knows it, and so on. 

A security is a function $X: \Omega \rightarrow {\mathbb R}$, where $X(\omega)$ is the security's payoff at state $\omega$.  Let $X(C(\omega)) = \underset{\omega' \in C(\omega)}{\bigcup} X(\omega')$ be the collection of all values of security $X$ that can be realised in states reachable from $\omega$. Intuitively, these are the values of $X$ such that some agent $i_1$ thinks that $i_2$ thinks that \ldots some $i_n$ considers possible. Alternatively, $X(C(\omega))$ is the smallest set of the security's values that it is common knowledge at $\omega$ that  agents consider possible. Let $\max X(C(\omega))$ be the maximum and $\min X(C(\omega))$ the minimum of these values.

When a state $\omega$ occurs, agent $i$ receives private information $\Pi_i(\omega)$ and updates his prior $p_i$. His expectation of $X$ at $\omega$ is therefore $e_i(\omega) \equiv \underset{\omega' \in \Pi_i(\omega)}{\sum} X(\omega') \frac{\pi_i(\omega')}{\pi_i(\Pi_i(\omega)}$. The event ``Agent \is expectation of $X$ is $e_i$" consists of all states $\omega$ such that $e_i(\omega) = e_i$. 

\subsection{Verifiability}

%In a decentralised environment, there are no trusted intermediaries that report the true value of $X$. Instead, the true value is verified by some agent,  called an oracle. However, an oracle can also participate in a trade, because the identities of all agents are hidden. If an oracle knows too much at the time the trade is agreed, the other agents may be unwilling to buy or sell. 

We  introduce four  properties on the private information of agents that verify the value of security $X$ and in the next section we examine their implications for trade. Note that we allow for different priors. Under a common prior, there would be no trade so verifibability would not have any implications.

\begin{Definition}
Security $X$ is verifiable if, for each $\omega \in \Omega$, there exists $i \in I$ such that $\Pi_i(\omega) \sq X^{-1}(k)$, for some $k \in {\mathbb R}$. 
%It is collectively verifiable if $\underset{i \in I}{\bigcap} \Pi_i(\omega) \sq X^{-1}(k)$.
\end{Definition}

Verifiability specifies that, at each state, there exists at least one agent  who knows the true value of the security. 
%By verifying the true value and reporting it on the decentralised network, it becomes common knowledge and the trades can be settled. For example, suppose there is an agreed trade between two agents, where $i$ buys 2 units of security $X$ from $j$. If the value of the security is $X(\omega) = \$5$ and there is an oracle $k$ who can verify its true value, so that $\Pi_k(\omega) \sq X^{-1}(5)$, then  agent $i$ receives $\$10$ from agent $j$. However, if the oracle has this information before the trade is agreed, he will try to participate in order to take advantage of either $i$ or $j$. 
This is a strong property that implies no common knowledge trade. Maxmin verifiability, which is weaker, specifies that among all the values of $X$ that it is common knowledge that they are possible, there is a agent who knows either the  highest or  lowest value of $X$, when it is true. For example, consider a security which is linked to the value of a new oil well. The possible values are between zero and some positive number, depending on how much oil can be extracted. An expert can determine the range of values by running tests if given access to the well, however if there is no oil at all then she will be able to verify this. Hence, if the expert is also a agent in the market, maxmin verifiability is satisfied, because when one extreme value is true, at least one agent knows it. 

Another example involves scientists engaged in the development of a breakthrough drug, either within a pharmaceutical firm or a regulatory authority overseeing its approval. While the success of the drug remains uncertain, the firm's valuation is constrained within a certain range, sustained by ongoing research efforts and investment. However, once the drug is either demonstrably successful or on the verge of regulatory approval, it becomes apparent that the firm's value will approach its maximum. If these scientists are permitted to participate in trading before the company's announcement, the condition of maxmin verifiability is fulfilled.

%Another example is that of scientists researching a new breakthrough drug within a pharmaceutical company, or  working for the regulator that approves the drug. As long as the drug is not a success, the value of the company is within a range, because there is funding to keep trying. However, once the drug is successful or it is going to be approved by the regulator, they know that the value of the company will reach its maximum. If these scientists are allowed to trade, then maxmin verifiability is satisfied.  

Finally, a catastrophe bond (CAT bond) provides high yields to investors who buy it but the entire principal is lost if a specific disaster, such as a hurricane, occurs before it matures. The issuance of CAT bonds is premised on the violation of maxmin verifiability, so that it is common knowledge that no scientist could predict with probability 1 that a hurricane of a certain magnitude will occur .%
\footnote{The maturity of CAT bonds is usually between 1 and 5 years. Longer maturities are probably infeasible as scientists would accurately predict that a catastrophic event will occur with probabiltiy close to 1, hence maxmin verifiability is satisfied and such bonds would not be traded.}
%

%among the possible payoffs of a security linked to the quarterly earnings of a company that it is common knowledge they could occur, there is a agent who knows that the maximum (or the minimum) has occurred, when the trade is agreed.  For instance, the worst performance for the company when there is an insurance claim This could happen because when the best case 

\begin{Definition}
Security $X$ is maxmin verifiable  at $\omega$ if  there exist $i \in I$ and $\omega' \in C(\omega)$ such that $\Pi_i(\omega') \sq X^{-1}(k)$, where $k \in \{\max X(C(\omega)), \min X(C(\omega))\}$.
\end{Definition}

This is a significant weakening of verifiability because only one value is verified, not all. Moreover, it is  possible that the true value $X(\omega)$ is different from both  $\max X(C(\omega))$ and  $\min X(C(\omega))$. Nevertheless, we show that this property is still strong enough to preclude common knowledge trade. The following property is weaker than maxmin verifiability.

% We say that security $X$ is maxmin verifiable at $\omega$ if some $i_1$ thinks that $i_2$ thinks that \ldots some $i_{n-1}$ thinks possible that $i_n$ knows the maximum or the minimum. 

%\begin{Definition}
%Security $X$ is threshold verifiable  at $\omega$ if  there exist $i \in I$ and $\omega' \in C(\omega)$ such that $\Pi_i(\omega') \sq X^{-1}(k)$, where $k \in \{\underset{\omega'' \in C(\omega) }{min} X(\omega''), \underset{\omega'' \in C(\omega)}{max} X(\omega'')\}$.
%\end{Definition}

\begin{Definition}
Security $X$ is threshold verifiable  at $\omega$ if, whenever $X(C(\omega))$ is not constant, there exist $i \in I$ and $\omega',\omega'' \in C(\omega)$ such that $\underset{\omega_0 \in \Pi_i(\omega')}{\max}X(\omega_0) < x < \underset{\omega_0 \in \Pi_i(\omega'')}{\min}X(\omega_0)$, for some $x \in {\mathbb R}$.
\end{Definition}

If security $X$ is threshold verifiable, there are two cases. First, $X(C(\omega))$ is constant, so its unique value is common knowledge at $\omega$. Second, $X(C(\omega))$ is not constant but
% it is threshold verifiable at $\omega$ if, among all the values of $X$ that it is common knowledge that they are possible, 
there is a threshold $x$ and a agent  who knows at some state whether the value of the security is strictly above or below $x$. It is important to emphasise that the agent knows whether the value is above or below the threshold at {\it some} state, not at all states.
%there is a agent  who can distinguish between the low and high values of $X$ whether the value of the security is above or below $x$. 
Theorem \ref{char} shows that threshold verifiability is necessary and sufficient for no common knowledge trade.

Examples for threshold verifiability are similar to the ones for maxmin verifiability, but the requirements are weaker. A scientist working in a pharmaceutical company knows whether its value is below or above a threshold, as he has inside information on how well the new drugs perform in tests. Similarly, an oil expert knows whether the value of the well will be above or below a threshold, depending on the outcome of the tests she conducts.

The final property specifies that the value of the security is verified only when all agents share their private information.

\begin{Definition}
Security $X$ is collectively verifiable if, for each $\omega \in \Omega$, $\underset{i \in I}{\bigcap} \Pi_i(\omega) \sq X^{-1}(k)$, for some $k \in {\mathbb R}$.
\end{Definition}

Collective verifiability is the minimum requirement which allows agents to agree a trade now and reveal their private information at a later date, so that the trade is settled and payments are made. We show with an example that this property is not strong enough to preclude common knowledge trade.

\section{Trading environments}
\label{trading environments}
We explore three different trading environments. The first is static, whereas the other two are dynamic.

\subsection{We cannot agree to disagree}
We say that there is common knowledge trade at $\omega$  if the agents' expectations about the value of $X$ are common knowledge at $\omega$, but different. This is interpreted as agreeing to trade with each other, fully taking into account that the others  are also willing to do so. For example, if $e_i > e_j$, then agent $i$ is willing to buy some units of $X$ from $j$. See Section \ref{multiple securities}, where we extend trading to a collection of securities $\{X_i\}_{i \in I}$ with $\underset{i \in I}{\sum} X_i = 0$. Recall that \cite{aum76} shows that common knowledge trade is impossible with a common prior, however in this setting priors can be different.
\begin{Definition}
There is common knowledge trade at $\omega$ if the event ``agent \is expectation of $X$ is $e_{i}$, for each $i \in I$'' is common knowledge at $\omega$, yet $e_{i} \neq e_{j}$ for some $i,j \in I$.
\end{Definition}

The following theorem shows that threshold verifiability is necessary and sufficient to preclude any common knowledge trade. To prove necessity, we make the additional assumption that the security pays differently across states, so that  $X(\omega) \neq X(\omega')$ for all $\omega, \omega' \in \Omega$. This precludes the uninteresting case where the intersection of the security's values that each agent considers possible across states, $\underset{\omega' \in C(\omega)}{\bigcap} [\underset{\omega_0 \in \Pi_i(\omega')}{\min}X(\omega_0), \underset{\omega_0 \in \Pi_i(\omega')}{\max}X(\omega_0) ]$, is a singleton and the same for all agents, so that there is no common knowledge trade for any set of priors.

\begin{Theorem} \label{char}
%If security $X$ is threshold verifiable at $\omega$,  then there cannot be common knowledge trade at $\omega$.
% %If at $\omega$ the agents' expectations about $X$ are common knowledge, then they are the same. 
% \red{this is a proof for threshold verifiability.} \red{Conversely, if $X$ is not threshold verifiable at $\omega$, then there exist prior beliefs   $\{p_i\}_{i \in I}$  such that there cannot be common knowledge trade at $\omega$. Conversely, if there is no common knowledge trade at $\omega$ for any prior beliefs and $X$ is not constant, then it is  threshold verifiable at $\omega$.}

%If security $X$ is threshold verifiable at $\omega$, then there is no common knowledge trade at state $\omega$, for any set of priors. Conversely,  if $X$ is not threshold verifiable at $\omega$, then there is common knowledge trade at $\omega$ and for some set of priors.
Security $X$ is threshold verifiable at $\omega$ if and only if there is no common knowledge trade at $\omega$, for any set of priors.
%and only if security $X$ is threshold verifiable at $\omega$.
\end{Theorem}

\begin{proof}

Suppose $X$ is threshold verifiable at $\omega$. If $X(C(\omega))$ is constant, everyone knows the security's value and there is no common knowledge trade. If  $X(C(\omega))$ is not constant, there exist $i \in I$ and $\omega',\omega'' \in C(\omega)$ such that $\underset{\omega_0 \in \Pi_i(\omega')}{\max}X(\omega_0) < \underset{\omega_0 \in \Pi_i(\omega'')}{\min}X(\omega_0)$. This implies that for any prior $p_i$, \is expected value of $X$ at $\omega'$ is different from \is  expected value of $X$ at $\omega''$. Because they are different, expected values cannot be common knowledge at $\omega$, hence there is no common knowledge trade at $\omega$. 

Conversely, suppose that $X$ is not threshold verifiable at $\omega$, hence $X(C(\omega))$ is not constant. We will show that there is common knowledge trade at $\omega$ for some set of priors. The negation of threshold verifiability implies that for each $i \in I$ we have $\underset{\omega' \in C(\omega)}{\bigcap} [m_i^{\omega'}, M_i^{\omega'} ] \neq \emptyset$, where $m_i^{\omega'} = \underset{\omega_0 \in \Pi_i(\omega')}{\min} X(\omega_0)$ is the minimum and $M_i^{\omega'} = \underset{\omega_0 \in \Pi_i(\omega')}{\min} X(\omega_0)$ is the maximum value of $X$ given \is partition cell at $\omega'$. We need to show that there is common knowledge trade at $\omega$ for some set of priors $\{p_i\}_{i \in I}$. For each $i \in I$, pick $k_i \in \underset{\omega' \in C(\omega)}{\bigcap} [m_i^{\omega'}, M_i^{\omega'} ] \neq \emptyset$. Because $X$ pays differently at each state, so that $X(\omega') \neq X(\omega'')$ for all $\omega', \omega'' \in \Omega$, we have that $\underset{\omega' \in C(\omega)}{\bigcap} [m_i^{\omega'}, M_i^{\omega'} ]$ is not a singleton. This implies that we can choose $k_i$ such that $k_i \in (m_i^{\omega'}, M_i^{\omega'} )$ and we can find $i,j \in I$ such that $k_i \neq k_j$.

We interpret $k_i$ as \is expected value of $X$, which is constant for all $\omega' \in C(\omega)$. Using the following procedure, we construct, for each $i \in I$, a prior $p_i$ that generates $k_i$ as the expected value of $X$, given each partition cell $\Pi_{i}(\omega')$, $\omega' \in C(\omega)$.  Given $\omega' \in C(\omega)$, we have that $k_i \in (m_i^{\omega'}, M_i^{\omega'} )$. This implies that we can find a posterior $p_i^{\omega'}$ with full support on $\Pi_i(\omega')$, such that $E_{p_i^{\omega'}}[X] = k_i$. This is true for all partition cells within $C(\omega)$. Let ${\cal P}^i_\omega$ be the collection of these generated posteriors.  By assigning  positive weights $\pi(\omega')$ to posteriors $p_i^{\omega'} \in {\cal P}^i_\omega$ that add up to 1, we construct the prior $p_i = \underset{p_i^{\omega'} \in {\cal P}^i_\omega}{\sum} \pi(\omega')p_i^{\omega'}$ for agent $i$. Note that there are infinitely many such priors. Each $k_i$ is constant across all partition cells in $C(\omega)$, hence the expected values of $X$ are common knowledge. Because $k_i \neq k_j$ for some $i,j \in I$, there is common knowledge trade $\omega$.

\end{proof}

The theorem specifies that if threshold verifiability fails, then there is common knowledge trade for some set of priors. However, the proof provides a stronger result, that there is common knowledge trade for infinitely many sets of priors. To provide some intuition, note that if threshold verifiability fails at $\omega$, then for each $i \in I$ we have $\underset{\omega' \in C(\omega)}{\bigcap} [m_i^{\omega'}, M_i^{\omega'} ] \neq \emptyset$, where $m_i^{\omega'} = \underset{\omega_0 \in \Pi_i(\omega')}{\min} X(\omega_0)$ is the minimum and $M_i^{\omega'} = \underset{\omega_0 \in \Pi_i(\omega')}{\min} X(\omega_0)$ is the maximum value of $X$ given \is partition cell at $\omega'$. For each $i \in I$, we can pick $k_i \in \underset{\omega' \in C(\omega)}{\bigcap} [m_i^{\omega'}, M_i^{\omega'} ] \neq \emptyset$, which is interpreted as agent \is expected value of $X$, constant across all his partition cells. If $k_i \neq k_j$ for at least two agents, then there is common knowledge trade. We can find infinitely many such vectors $\{k_i\}_{i \in I}$ by choosing different points in the intervals $\underset{\omega' \in C(\omega)}{\bigcap} [m_i^{\omega'}, M_i^{\omega'} ] \neq \emptyset, i \in I$. Moreover, each such vector $\{k_i\}_{i \in I}$ can be generated by infinitely many sets of priors, as we describe in the proof.

Note that maxmin verifiability implies no common knowledge trade, as it is stronger than threshold verifiability. The following example shows that if the security is collectively verifiable, then there can be common knowledge trade. 

\begin{Example} \label{collectively verifiable}
Let the state space be $\Omega = \{\omega_1, \omega_2, \omega_3, \omega_4\}$ and consider two agents. Agent 1's partition is $\Pi_1 = \{\{\omega_1, \omega_2\}, \{\omega_3, \omega_4\}\}$, whereas 2's partition is $\Pi_2 = \{\{\omega_1, \omega_3\}, \{\omega_2, \omega_4\}\}$. Agent 1's prior is  $p_1 = \{1/6,1/3,1/3,1/6\}$ and agent 2's is $p_2 = \{1/3,1/6,1/6,1/3\}$. The security is $X(\omega_1) = X(\omega_4) = 1$ and $X(\omega_2) = X(\omega_3) = -1$. It is collectively verifiable because at each state, the agents' collectively know the value of $X$. However,  at any state $\omega$ it is common knowledge that 1's expectation of $X$ is $-1/3$, whereas 2's expectation is $1/3$. This means that there is common knowledge trade, where $2$ buys the security from $1$. 
\end{Example}

\subsection{We cannot disagree forever}

%In a static environment, we showed that threshold verifiability is sufficient to preclude common knowledge trade. 
We now analyse a dynamic trading environment, based on \cite{geaPol82}. In every period, one agent announces his expected value of the security $X$. 
%which models agents taking turns in announcing their posterior belief about an event $E \sq \Omega$. 
Each announcement reveals some information to the other agents, who then update their own posterior beliefs. \cite{geaPol82} show that if the agents share a common prior, they will eventually agree on their expected value of $X$. If priors are different, however, agents may disagree forever.

%Our setting is similar, but differs in the following aspects. First, there is no common prior. Second, each agent announces his expected value of the security, instead of the posterior belief about an event. As with the static environment, the interpretation is that if the expectations are different, there is scope for trade. The trades are settled with information provided by the oracles. However, what happens if some of the oracles are able to disguise themselves as agents? This is facilitated by the decentralised nature of the blockchain, where the identity of the agents is hidden.

Formally, there are infinitely many periods $t = 0,1, 2, \ldots$. At period $t=0$ and state $\omega$, agent \is partition cell is denoted $\Pi^0_i(\omega)$ and the public information created by the announcement is $C^0(\omega) = \Omega$. Agents make announcements sequentially. At  period $t$,  agent $j$ announces his expected value of security $X$, according to his prior $p_j$ and  partition cell $\Pi_j^{t-1}(\omega) = \Pi_j^0(\omega) \cap C^{t-1}(\omega)$, which is formulated by his initial private   information $\Pi_j^0$ and the public information $C^{t-1}(\omega)$ that has been revealed by all previous announcements. The public information, created by \js announcement $e_t$ at period $t$, is defined as $C^{t}(\omega) = \{\omega' \in C^{t-1}(\omega): \underset{\omega'' \in \Pi_j^{t-1}(\omega')}{\sum} X(\omega'')\frac{p_j(\omega'')}{p_j(\Pi_j^{t-1}(\omega'))} = e_t\}$, the set of states which are consistent with all announcements up to $t$.

After each announcement, all  other agents update their own information, by excluding any states that would not result in agent $j$ making this announcement. Because the state space is finite, this process of updating of information will eventually stop. Suppose that this happens in period $t^*$. This means that each agent $j$ will make the same announcement $e_{j}$ at all states in $C^{t^*}(\omega)$.%
\footnote{If he was not making the same announcement, then at some $t > t^*$ the public information $C^{t^*}(\omega)$ would further shrink, contradicting that the updating of information has stopped.}
If they agree on their announcements, we say that they cannot disagree forever.

\begin{Definition}
Agents cannot disagree forever at $\omega$ if at $t^{*}$ we have $e_{i} = e_{j}$ for all $i,j \in I$.
\end{Definition}

Is threshold verifiability of $X$ at $t=0$ and $\omega$, or even maxmin verifiability, sufficient for no disagreement? 
%\cite{geaPol82} showed that a common prior implies that agents will eventually agree, so that they will end up announcing the same expectation of $X$.%
%\footnote{\cite{geaPol82} only analysed the case of announcing the posterior of an event $E$, however this is easily generalisable to the expectation of a random variable.}
%%
We  show by example that it is not. 
\begin{Example} \label{partial dynamic}
Let the state space be $\Omega = \{\omega_1, \omega_2, \omega_3, \omega_4, \omega_5\}$ and consider two agents. Agent 1's partition is $\Pi_1 = \{\{\omega_1, \omega_2\}, \{\omega_3, \omega_4\}, \{\omega_5\}\}$, whereas 2's partition is $\Pi_2 = \{\{\omega_1, \omega_3\}, \{\omega_2, \omega_4, \omega_5\}\}$. Agent 1's prior is  $p_1 = \{1/12,1/6,1/6,1/2, 1/2\}$ and agent 2's is $p_2 = \{1/6,1/12,1/12,1/6, 1/2\}$. The security is $X(\omega_1) = X(\omega_4) = 1$, $X(\omega_2) = X(\omega_3) = -1$ and $X(\omega_5) = 5$. The security is maxmin, and therefore partially, verifiable because at any state $\omega$, $C(\omega) = \Omega$ and agent 1 knows the maximum value of $X$ at $\omega_{5}$.
%$\underset{\omega_0 \in \Pi_1(\omega_{4})}{\max}X(\omega_0) < \underset{\omega_0 \in \Pi_1(\omega_{5})}{\min}X(\omega_0)$.

From  Theorem \ref{char}, we know that there  is no common knowledge trade at $t=0$. However, it is straightforward to show that there can be disagreement and trade in a dynamic setting. If the true state is $\omega \neq \omega_5$, agent 1 does not announce $5$, which informs agent 2 that the state is not $\omega_5$.  By excluding state $\omega_5$, the  partitions and the posteriors are updated  so that the example becomes identical to Example \ref{collectively verifiable}. Then, agent 1 announces 1/3 and agent 2 announces -1/3, so that there is no more updating and there is common knowledge trade and disagreement forever. 
\end{Example}

Disagreement and trade are possible because maxmin verifiability in period $t$ does not imply maxmin verifiability at $t+1$. The reason is that even if the maximum (or minimum) value of $X(C^t(\omega))$ is verified, this may not be true for  $X(C^{t+1}(\omega))$. In Example \ref{partial dynamic}, the maximum value 5 is verified at $t=0$ but at $t=1$ neither the maximum 1,  nor the minimum -1, are verified. In contrast, a common prior at $t$ continues to be common at $t+1$, as long as all agents use Bayes' rule, hence the result of \cite{geaPol82}, that agents cannot disagree forever. In order to preclude trade in our environment, we need threshold verifiability at period $t^*$, which is the period after which there is no more updating of information by anyone. Because at $t^*$ the announcement of each agent is constant in $C^{t^*}(\omega)$, threshold verifiability implies that $X(C^{t^*}(\omega))$ is constant. We therefore have the following Corollary.

\begin{Corollary}
If at $\omega$ and $t^*$ security  $X$ is threshold verifiable, then agents cannot disagree forever at $\omega$.
\end{Corollary}

\subsection{Information aggregation}

The question of no trade is closely related to the question of whether a security can aggregate information, so that the price always converges to its true value.
%This means that the fluctuating price of the traded security eventually reveals the private information of the agents, thus converging to a price which is equal to the true value of the security.
In this section we model trading using the market scoring rule, which is also used in prediction markets. Trading is organised as follows. At  $t = 0$, nature selects a state $\omega \in \Omega$ and the uninformed market maker makes a prediction $y_0$ about the value of security $X: \Omega \rightarrow {\mathbb R}$. At  $t = 1$, agent 1 makes a revised prediction $y_1$, at $t = 2$ agent $t_2$ makes his prediction, and so on. At  $t = n+1$, agent 1 makes another prediction $y_{n+1}$. Each prediction $y_k$ is required to be within the set $Y=[\underset{\omega\in \Omega}\min X(\omega), \underset{\omega \in \Omega}\max X(\omega)]$.

%Let $a(t)$ be the agent that makes a prediction at time $t$. All predictions are observed by all agents. 

%The process repeats until  $t_\infty = \lim_{k \rightarrow \infty} t_k$. At  $t^*>t_\infty$, 

The agents' payoffs are computed using a scoring rule, $s(y,x^*)$, where $x^*$  is the true value of the security  and $y$ is a prediction. A scoring rule is {\it proper} if, for any probability measure $p$ and any random variable $X$, the expectation of $s$  is maximised at $y=E_{p}[X]$. It is  {\it strictly proper} if $y$ is unique. We focus on continuous strictly proper scoring rules. Examples are the quadratic, where $s(y,x)=-(x-y)^{2}$, and the logarithmic, where $s(y,x)=(x-a)ln(y-a)+(b-x)ln(b-y)$ with $a<\underset{\omega\in \Omega}\min X(\omega), b>\underset{\omega\in \Omega}\max X(\omega)$. A agent's payoff from announcing $y_t$ at $t$,   is $s(y_{t},x^*)-s(y_{t-1},x^*)$, where  $y_{t-1}$ is the previous announcement and $x^*$ is the true value of the security. 

We assume that each agent is myopic, so that he cares only about the current payoff, when making an announcement. Because scoring rules are strictly proper and agents are myopic, their announcement will always be their expected value of $X$. Hence, this model is closely related to the model of the previous section.  
%This implies that each agent announces his expected value of the security.  
We say that information gets aggregated if the agents' predictions converge to the true value of the security, $X(\omega)$.

\begin{Definition} \label{def info agg}
Information gets aggregated at $\omega$ if sequence $\{y_k\}_{k=1}^{\infty}$ converges in probability to random variable $X(\omega)$. 
\end{Definition}

\cite{ostrovsky12} and \cite{chenEtAl12} show that, with common priors, information gets aggregated at all states if and only if the security is separable.%
\footnote{\cite{ostrovsky12} characterises this class of separable securities, which includes the Arrow-Debreu securities.}
%  \cite{galanisKotronis21} and \cite{galanisEtAl24} examine information aggregation in environments with unawareness and ambiguity aversion, respectively.}
%
This means that if the security is not separable, it is possible that agents agree on its expected value, but this is different from its true value at $\omega$.

In the previous section, we showed that threshold verifiability at $t^{*}$ is strong enough to imply agreement. Moreover, it implies that $X(C^{t^*}(\omega))$ is constant. Because there is only one value of $X$ that is possible, it must be the correct one and we have information aggregation. It is important to note that this result holds for all securities, not just the separable ones.

\begin{Corollary} \label{partial thm}
If  security $X$ is threshold verifiable at $t^*$ and $\omega$, then there is information aggregation at $\omega$.
\end{Corollary}

\subsection{Multiple securities}
\label{multiple securities}

In the current model, only one security, $X$,  is traded. In this section, we discuss how trading can be generalised to multiple securities, $\{X_i\}_{i \in I}$, one for each agent, so that $\underset{i \in I}{\sum}X_i = 0$. First, we fix the set of priors and  show that common knowledge trade with $X$ implies common knowledge trade with $\{X_i\}_{i \in I}$. Suppose, without loss of generality, that the only common knowledge event is the whole state space $\Omega$.  Then, our definition of common knowledge trade with $X$ is that each agent \is expected value given each of his partition cells is $e_i$, yet $e_i \neq e_j$ for at least two agents. If $e_i > e_j$, then Agent $i$ is willing to  buy the security from Agent $j$ at some intermediate price $e_i > p > e_j$. We can model this trade by having two new securities, $X_i = X-{\bf p}$ and $X_j = -X + {\bf p}$, where ${\bf p}$ is the security that pays $p$ at all states, so we have $X_i+X_j = {\bf 0}$. Both agents have strictly positive expected value from their respective securities, and this trade is feasible because the two securities add up to 0 at all states. 

The two approaches are complementary. We can view $X$ as the `fundamental' asset which pays according to the state (e.g. the value of the oil well), whereas securities $X_i, X_j$ denote the actual trade that has been agreed. With $n$ agents, where $n$ is even, we can order $e_1 > e_2 > \ldots > e_n>0$. By setting $p$ such that $e_{n/2} > p > e_{n/2 + 1}$, we can define $X_i = X - {\bf p}$ for $i \leq n/2$ and $X_i = -X + {\bf p}$ otherwise. We then have common knowledge trade with multiple securities. The converse is not necessarily true. To see this, suppose that $\{X_i\}_{i \in I}$ satisfy common knowledge trade with multiple securities. We can define $X = \underset{j \in I}{\sum} \lambda_j X_i$, so we need to find $\lambda = \{\lambda_1, \ldots, \lambda_n\}$ such that $\underset{j \in I}{\sum} \lambda_j E_{p_i}[ X_j | \Pi_i(\omega)] = e_i$, for all $\omega \in \Omega$ and $i,j \in I$. As there are many more equations than unknowns, the system may not have a solution.

We now allow for all possible priors and extend our analysis to multiple securities. We  say that a collection $\{X_i\}_{i \in I}$ of securities is tradable if they add up to zero at all states, $\underset{i \in I}{\sum}X_i = 0$, so that the trade is feasible, and every Agent \is maximum value of $X_i$ given each of his partition cells is strictly positive. If the last condition was not satisfied for some partition cell, then the expectation of $X_i$ would be negative for all priors, hence Agent $i$ would not agree to participate in the trade.

\begin{Definition}
Securities $\{X_i\}_{i \in I}$ are tradable if $\underset{i \in I}{\sum}X_i = 0$ and $\underset{\omega_0 \in \Pi_i(\omega)}{\max} X(\omega_0) > 0$, for all $i \in I$ and $\omega \in \Omega$.
\end{Definition}

We say that there is common knowledge trade given a collection of tradable securities if it is common knowledge that each Agent \is expected value of $X_i$ is $e_i>0$. 

\begin{Definition}
There is common knowledge trade at $\omega$ (with tradable securities $\{X_i\}_{i \in I}$)  if the event ``agent \is expectation of $X_i$ is $e_i >0$, for all $i \in I$'' is common knowledge at $\omega$. 
\end{Definition}

Note that $e_i$ in this definition denotes the expected profits of Agent $i$, whereas $e_i$ in the original definition denotes \is expected value of $X$. \cite{samet98} uses a weaker definition, which only requires that the expectations are strictly positive. He shows that there is a common prior that generates a  given collection of posterior beliefs, one for each agent and his partition cells, if and only if there is no common knowledge trade with multiple securities.

%Common knowledge trade with $X$ implies common knowledge trade with multiple securities. To see this, suppose $e_1 > e_2 > \ldots > e_n>0$ and $n$ is even. By setting $p$ such that $e_{n/2} > p > e_{n/2 + 1}$, we can define $X_i = X - {\bf p}$ for $i \leq n/2$ and $X_i = -X + {\bf p}$ otherwise. We then have common knowledge trade with multiple securities. 

We extend threshold verifiability  in a similar way.

\begin{Definition}
Tradable securities $\{X_i\}_{i \in I}$ are threshold verifiable  at $\omega$ if, whenever $X_i(C(\omega))$ is not constant for some $i \in I$, there exist $i \in I$ and $\omega',\omega'' \in C(\omega)$ such that $\underset{\omega_0 \in \Pi_i(\omega')}{\max}X_i(\omega_0) < x < \underset{\omega_0 \in \Pi_i(\omega'')}{\min}X_i(\omega_0)$, for some $x \in {\mathbb R}$.
\end{Definition}

We finally have the following Proposition, which is an extension of Theorem \ref{char}. The proof is almost identical so it is omitted.%
\footnote{Note that we do not require that each $X_i$ pays differently across states, as in the proof of Theorem \ref{char}. The reason is that $e_i = e_j$ is allowed here, as they are the expected profits, not the agents' expected value of $X$.}

\begin{Proposition}
Tradable securities $\{X_i\}_{i \in I}$ are threshold verifiable at $\omega$ if and only if there is no common knowledge trade at $\omega$, for any set of priors.
\end{Proposition}

\bibliographystyle{apalike}

\bibliography{master}

\end{document}